\documentclass{amsproc}

\usepackage{amssymb,amsmath,amscd,dsfont}
\newcommand{\op}[1]{\ensuremath{\operatorname{#1}}}
\newcommand{\wt}[1]{\ensuremath{\widetilde{#1}}}
\newcommand{\wh}[1]{\ensuremath{\widehat{#1}}}
\newcommand{\ol}[1]{\ensuremath{\overline{#1}}}
\newcommand{\ul}[1]{\ensuremath{\underline{#1}}}
\newcommand{\cK}{\ensuremath{\mathcal{K}}}
\newcommand{\fz}{\ensuremath{\mathfrak{z}}}
\newcommand{\fk}{\ensuremath{\mathfrak{k}}}
\newcommand{\fg}{\ensuremath{\mathfrak{g}}}
\newcommand{\Q}{\ensuremath{\mathds{Q}}}
\newcommand{\R}{\ensuremath{\mathds{R}}}
\newcommand{\N}{\ensuremath{\mathds{N}}}
\newcommand{\Z}{\ensuremath{\mathds{Z}}}
\newcommand{\id}{\ensuremath{\operatorname{id}}}
\newcommand{\pr}{\ensuremath{\operatorname{pr}}}
\newcommand{\ev}{\ensuremath{\operatorname{ev}}}
\newcommand{\Hom}{\ensuremath{\operatorname{Hom}}}
\newcommand{\tx}[1]{\ensuremath{\text{#1}}}
\newcommand{\SU}{\ensuremath{\operatorname{SU}}}
\newcommand{\per}{\ensuremath{\operatorname{per}}}
\newcommand{\se}{\ensuremath{\nobreak\subseteq\nobreak}}
\newcommand{\from}{\ensuremath{\nobreak:\nobreak}}
\renewcommand{\to}{\ensuremath{\nobreak\rightarrow\nobreak}}
\newcommand{\sprod}{\ensuremath{\mathopen{\langle}\mathinner{\cdot},\mathinner{\cdot}\mathclose{\rangle}}}
\newcommand{\rrarrow}{\hspace{.05cm}\mbox{\put(0,-2){
$\rightarrow$}\put(0,2){ $\rightarrow$}\hspace{.5cm}}}

\usepackage[notcite,notref,final]{showkeys}
\usepackage{hyperref}

\renewcommand{\rrarrow}[1][]{\mbox{\put(0,-2){
$\rightarrow$}\put(0,2){ $\rightarrow$}\hspace{.5cm}$_{#1}$\hspace{.1cm}}}

\hypersetup{
    linktocpage = true,
    pdftitle = {Non-integral central extensions of loop groups},
    pdfauthor = {Christoph Wockel},
    bookmarksopen = true,
    bookmarksopenlevel = 1,
	unicode = true
  }

\newtheorem{theorem}{Theorem}[section]
\newtheorem{proposition}[theorem]{Proposition}
\newtheorem{corollary}[theorem]{Corollary}
\newtheorem{lemma}[theorem]{Lemma}

\theoremstyle{definition}

\theoremstyle{remark}
\newtheorem{remark}[theorem]{Remark}

\numberwithin{equation}{section}

\begin{document}
\sloppy

\title{Non-integral central extensions of loop groups}

\author{Christoph Wockel}
\address{Department Mathematik, Universit{\"a}t Hamburg, Bundesstr. 55, 20146 Hamburg, Germany}
\email{christoph@wockel.eu}

\date{January 1, 1994 and, in revised form, June 22, 1994.}


\begin{abstract}
It is well-known that the central extensions of the loop group of a
compact, simple and 1-connected Lie group are parametrised by their level
$k\in \Z$. This article concerns the question how much can be said for
arbitrary $k\in \R$ and we show that for each $k$ there exists a Lie
groupoid which has the level $k$ central extension as its quotient
if $k\in \Z$. By considering categorified principal bundles
we show, moreover, that the corresponding Lie groupoid has
the expected bundle structure.
\end{abstract}

\maketitle

\section*{Introduction}

In this paper we generalise a construction of the universal central extension $\wh{\Omega
K}$ of the loop group $\Omega K$ of a compact simple and 1-connected Lie group $K$, going
back to Mickelsson and Murray
\cite{Mickelsson87Kac-Moody-groups-topology-of-the-Dirac-determinant-bundle-and-fermionization},
\cite{Murray88Another-construction-of-the-central-extension-of-the-loop-group}. They
construct $\wh{\Omega K}$ as a quotient of a central extension $U(1)\times _{\kappa}P_{e}\Omega K$ of the based path group
$P_{e}\Omega K$. For this construction one has the freedom to
choose a real number $k\in \R$ (after having fixed all normalisations appropriately),
which is usually referred to as the \emph{level}. The construction from \cite{Mickelsson87Kac-Moody-groups-topology-of-the-Dirac-determinant-bundle-and-fermionization} and
\cite{Murray88Another-construction-of-the-central-extension-of-the-loop-group} then
yields a normal subgroup $N_{k}\unlhd U(1)\times _{\kappa}P_{e}\Omega K$ if
and only if $k\in \Z$ and constructs $\wh{\Omega K}$ as $(U(1)\times _{\kappa}P_{e}\Omega K)/N_{1}$.

The point of this article is that, although the construction of $N_{k}$ works if and \emph{only} if $k\in \Z$, for general $k$ there still exists an infinite-dimensional Lie group
$\cK$ acting on $U(1)\times _{\kappa}P_{e}\Omega K$ and the quotient of this
action coincides with $(U(1)\times _{\kappa}P_{e}\Omega K)/N_{k}$ if $k\in \Z$.
This then gives rise to an action Lie groupoid. By passing to a Morita equivalent
Lie groupoid we show that this Lie groupoid has the structure of a generalised principal bundle.

The results that we get here are closely related to the general extension theory of
infinite-dimensional Lie groups by categorical Lie groups from
\cite{Wockel08Categorified-central-extensions-etale-Lie-2-groups-and-Lies-Third-Theorem-for-locally-exponential-Lie-algebras}.
However, this article concerns more the global and differential point of view to those
extensions for the particular case of loop groups, while
\cite{Wockel08Categorified-central-extensions-etale-Lie-2-groups-and-Lies-Third-Theorem-for-locally-exponential-Lie-algebras}
provides a more detailed perspective from the side of cocycles.

\textbf{Notation:} Throughout this article, $G$ denotes a (possibly infinite-dimensional) connected Lie group with Lie algebra $\fg$, modelled on a locally convex space, $\fz$ is a sequentially complete locally convex space and $\Gamma\se \fz$ is a discrete subgroup of (the additive group of) $\fz$. Moreover, we set $Z:= \fz/\Gamma$.

Most of the time, $G$ will be the pointed loop group
\begin{equation*}
\Omega K:=\{\gamma\from \R\to K:\gamma(0)=e, \gamma(x+n)=\gamma(x)\,\forall n\in \N\}
\end{equation*}
of smooth and pointed loops in a compact, simple and 1-connected Lie group $K$, endowed with the usual Fr\'echet topology and point-wise multiplication.
The Lie algebra of $\Omega K$ is then \begin{equation*}
\Omega \fk:=\{\gamma\from \R\to \fk:\gamma(0)=0, \gamma(x+n)=\gamma(x)\,\forall n\in \N\}
\end{equation*}
with $\fk:= L(K)$. In this case, $\fz$ will be $\R$, $\Gamma$ will be $\Z$ and thus
\mbox{$Z=\R/\Z=:U(1)$}. We circumvent all normalisation issues by choosing this quite unnatural
realisation of the circle group. Moreover, we denote by $\exp$ the canonical quotient map
\mbox{$\exp:\R\to\R/\Z$}.

We will be a bit sloppy in our conventions concerning the precise model for $S^{1}$ and $B^{2}$. Instead, we collect the things that we want to assume:
\begin{itemize}
 \item $B^{2}$ and $S^{1}$ are manifolds with corners such that $S^{1}$
       may be identified with a submanifold of $B^{2}$ (which we denote
       by $\partial B^{2}$) and the base-point of $B^{2}$ is contained
       in $\partial B^{2}$.
 \item $C^{\infty}_{*}(S^{1},G)$ may be identified with the kernel of
       the evaluation map $\ev\from P_{e}G\to G$, where $P_{e}G$ denotes
       the space of smooth maps $f\from [0,1]\to G$ with $f(0)=e$.
 \item The map $C_{*}^{\infty}(B^{2},G)\to C_{*}^{\infty}(S^{1},G)_{e}$,
       $f\mapsto \left.f\right|_{\partial B^{2}}$ is surjective.
\end{itemize}
Here, the subscript $_{*}$ denotes pointed maps and the subscript $_{e}$ denoted the
connected component of the identity.

\section{Generalities on  central extensions of infinite-dimensional Lie groups}
\label{sect:generalitiesOnCentralExtensions}

We briefly review essentials on central extensions of infinite-dimensional Lie group, established by Neeb in \cite{Neeb02Central-extensions-of-infinite-dimensional-Lie-groups}.
There, the second locally smooth group cohomology $H^{2}_{\op{loc}}(G,Z)$ is defined to be the set of functions $ f\from G\times G \to Z$ such that
\begin{itemize}
 \item $f$ is smooth on $U\times U$ for $U\se G$ some open identity neighbourhood
 \item $ f(g,h)+f(gh,k)=f(g,hk)+f(h,k)$ for all $g,h,k\in G$
 \item $f(g,e)=f(e,g)=0$ for all $g\in G$,
\end{itemize}
(called \emph{locally smooth group cocycles} in this paper) modulo the equivalence relation
\begin{equation}\label{eqn:cocycleCoboundaryaCondition}
(f\sim f'):\Leftrightarrow f(g,h)-f'(g,h)=b(g)-b(gh)+b(h)
\end{equation}
for some $g\from G\to Z$ which is smooth on some identity neighbourhood and
satisfies $b(e)=0$. Similarly, we define $H^{2}_{\op{glob}}(G,Z)$ to be defined
in the same way except that we require $f$ (respectively $b$) to be smooth on
$G\times G$ (respectively $G$). We shall call such a $f$ a \emph{globally smooth
group cocycle}. Then in
\cite{Neeb02Central-extensions-of-infinite-dimensional-Lie-groups} it is shown
that $H^{2}_{\op{loc}}(G,Z)$ corresponds bijectively to the equivalence classes
of central extensions
\begin{equation}\label{eqn:centExt1}
 Z\to \wh{G}\to G
\end{equation}
of Lie groups such that \eqref{eqn:centExt1} is a \emph{locally} trivial principal bundle
and that $H^{2}_{\op{glob}}(G,Z)$ corresponds bijectively to equivalence classes of
central extensions of Lie groups such that \eqref{eqn:centExt1} is a \emph{globally}
trivial principal bundle.

The bulk of the work in
\cite{Neeb02Central-extensions-of-infinite-dimensional-Lie-groups} concerns the
integration issue for central extensions, i.e., how to derive a continuous Lie
algebra cocycle $D(f)\from \fg \times\fg\to \fz$ from a locally smooth group
cocycle and to determine whether for a given continuous Lie algebra cocycle
$\omega$ there exists a locally smooth group cocycle $f$ such that
$[Df]=[\omega]\in H^{2}_{c}(\fg,\fz)$ (where the subscript $_{c}$ means continuous Lie
algebra cohomology). In the latter case we say that $\omega$ \emph{integrates}. The main
result in \cite{Neeb02Central-extensions-of-infinite-dimensional-Lie-groups} is an exact
sequence
\begin{multline*}
 \Hom(\pi_{1}(G),Z)\to H^{2}(G,Z)\xrightarrow{D}H^{2}_{c}(\fg,\fz)\\\xrightarrow{P}\Hom(\pi_{2}(G),Z)\oplus \Hom(\pi_{1}(G),\op{Lin}_{c}(\fg,\fz)),
\end{multline*}
where $\op{Lin}_{c}$ denotes continuous linear maps and
\begin{equation*}
 \per_{\omega}:=\pr_{1}( P([\omega]))\from \pi_{2}(G)\to Z,\quad [\sigma]\mapsto\left[ \int_{\sigma}\omega^{l}\right]
\end{equation*}
for $\sigma$ a smooth representative of $[\sigma]\in \pi_{2}(G) $ and $\omega^{l}$ the left-invariant 2-form on $G$ with $\omega^{l}(e)=\omega$. In particular, when $G$ is simply connected,
then the sequence reduces to a shorter exact sequence
\begin{equation*}
 0\to H^{2}(G,Z)\xrightarrow{D}H^{2}_{c}(\fg,\fz)\xrightarrow{\per}\Hom(\pi_{2}(G),Z). 
\end{equation*}
Thus a given cocycle $\omega$ integrates in this case if and only if the corresponding
\emph{period homomorphism} $\per_{\omega}$ vanishes.

\section{The topological type of central extensions}
\label{sect:theTopologicalTypeOfCentralExtensions}

In \cite{Mickelsson85Two-cocycle-of-a-Kac-Moody-group}, Mickelsson derives
a \v{C}ech 1-cocycle for $\wh{\Omega \SU_{2}}$.
In this section we shall describe how to derive the topological type of the 
principal bundle
\begin{equation*}
 Z\to \wh{G}\to G
\end{equation*}
for a central extension coming from a locally smooth cocycle $f\from G\times
G\to Z$. This description is much more general than the one from
\cite{Mickelsson85Two-cocycle-of-a-Kac-Moody-group} and it will become apparent
from this construction that for a globally smooth cocycle the corresponding
bundle is automatically trivial. For this we will make use of the following
fact.

\begin{theorem}
 \label{thm:globalisation-of-smooth-structures-on-groups}
 Let $H$ be a group $W\se H$ be a subset containing $e$ and let $W$
 be endowed with a manifold structure. Moreover, assume that there
 exists an open neighbourhood $Q\se W$ of $e$ with $Q^{-1}=Q$ and $Q\cdot Q\se W$ such that
 \begin{itemize}
  \item $Q\times Q\ni (g,h)\mapsto gh \in W$ is smooth,
  \item $Q\ni g\mapsto g^{-1}\in Q$ is smooth and
  \item $Q$ generates $H$ as a group.
 \end{itemize}
 Then there exists a manifold structure on $H$ such that $Q$ is open in
 $H$ and such that group multiplication and inversion is smooth.
 Moreover, for each other choice of $Q$, satisfying the above
 conditions, the resulting smooth structures on $H$ coincide.
\end{theorem}

\begin{proof}
	The proof is well-known and straight-forward, cf.\ \cite[Thm.\ II.1]{Wockel08Categorified-central-extensions-etale-Lie-2-groups-and-Lies-Third-Theorem-for-locally-exponential-Lie-algebras}, \cite[Prop.\ III.1.9.18]{Bourbaki98Lie-groups-and-Lie-algebras.-Chapters-1--3}.
\end{proof}

We now derive a central extension from a locally smooth group cocycle
$f\from G\times G\to Z$.
First, we define a twisted group structure
on the set-theoretical direct product $Z\times G$ by $(a,g)\cdot (b,h):=(a+b+f(g,h),gh)$. Then the requirement on $f$ to define a group
cocycle implies that this defines a group multiplication with neutral element $(0,e)$ and $(a,g)^{-1}=(-a-f(g,g^{-1}),g^{-1})$. We denote this group by $Z\times _{f}G$. If $f$ is smooth on
$U\times U$ and $V\se U$ is an open identity neighbourhood with $V\cdot V\se W$ and
$V^{-1}=V$, then $Z\times U$ carries the product manifold
structure and $Z\times V $ is open in $Z\times U$. Since $G$ is assumed to be connected, $Z\times_{f} G$ is generated by $Z\times V$ and the preceding theorem yields a Lie group structure on $Z\times _{f}G$. Clearly, the sequence
\begin{equation*}
 Z\to Z\times_{f}G\to G
\end{equation*}
is a locally trivial principal bundle for we have the smooth section
\begin{equation*}
 U\ni x\mapsto (0,x)\in Z\times U\se Z\times_{f}G.
\end{equation*}

\begin{lemma}
 The assignment
 \begin{equation*}
 \tau(f)_{g,h}\from  gV \cap hV \to Z,\quad x\mapsto f(g,g^{-1}x)-f(h,h^{-1}x)
 \end{equation*}
 defines a \v{C}ech 1-cocycle $\tau(f)$ on the open cover $(gV)_{g\in G}$ and thus an element of $\check{Z}^{1}(G,Z)$. If $[f]=[f']$ in $H^{2}_{\op{loc}}(G,Z)$,
then $[\tau(f)]=[\tau(f')]$  in $\check{H}^{1}(G,Z)$.
\end{lemma}

\begin{proof}
 We first note that $gh^{-1}\in V\cdot V\se W$ if $gV\cap hV\neq \emptyset$.
 From this it follows that 
 \begin{equation*}
  x\mapsto f(g,g^{-1}x)-f(h,h^{-1}x)=f(g^{-1}h,h^{-1}x)-f(g,g^{-1}h)
 \end{equation*}
is smooth on $gV\cap hV$, for $f$ is smooth on $U\times U$. From the
definition it is also clear that
$\tau(f)_{g,h}-\tau(f)_{g,k}+\tau(f)_{h,k}$ vanishes.

If $[f]=[f']$, then $f(g,h)-f'(g,h)=b(g)-b(gh)+b(h)$ for $b\from G\to Z$.
We assume without loss of generality that $f$ and $f'$ are smooth on $U\times U$ and $b$ is smooth on $U$ (presumably, the identity neighbourhoods may be distinct for $f$ and $f'$ and $b$). Then
\begin{equation*}
\tau(b)_{g}\from gV\to Z,\quad x\mapsto b(g)+b(g^{-1}x)
\end{equation*}
defines a {C}ech cochain with $\tau(f)-\tau(f')= \check{\delta}(\tau(b))$.
\end{proof}

We thus have a map $\tau\from H^{2}_{\op{loc}}(G,Z)\to \check{H}^{1}(G,Z)$, which clearly is a group homomorphism.

\begin{proposition}
	The principal bundle
	\begin{equation}\label{eqn:centExt2}
	 Z\to Z\times _{f}G\to G
	\end{equation}
	is classified by $[\tau(f)]\in\check{H}^{1}(G,Z)$.
\end{proposition}

\begin{proof}
 From the construction of the topology on $Z\times _{f}G$ it follows immediately
 that $\sigma_{e}(x):=(0,x)$ defines a smooth section on $V$. Thus the assignment
 \begin{equation*}
  \sigma_{g}\from gV\to Z\times_{f}G,\quad x\mapsto (f(g,g^{-1}x),x)=\left(\lambda_{(0,g)}\circ\sigma_{e}\circ\lambda_{{g}^{-1}}\right)(x)
 \end{equation*}
is smooth, where $\lambda_{g^{-1}}$ denotes left multiplication in $G$
with $g^{-1}$ and $\lambda_{(0,g)}$ denotes left multiplication in
$Z\times_{f} G$ with $(0,g)$. Consequently,
$(\sigma_{g}\from gV\to Z\times_{f}G)_{g\in G}$ defines a system of
sections for the principal bundle \eqref{eqn:centExt2} and since
$\tau(f)$ satisfies $\sigma_{g}(x)=\sigma_{h}(x)\cdot \tau(f)_{g,h}(x)$, this already shows the claim.
\end{proof}

\begin{corollary}\label{cor:globallySmoothCocycle}
A locally smooth cocycle $f\from G\times G\to Z$ is equivalent to a globally smooth cocycle if and only if the principal bundle, underlying $Z\to Z\times_{f}G\to G$, is topologically trivial.
\end{corollary}

\begin{proof}
If the bundle is topologically trivial, then there exists a smooth section
$\sigma\from G\to Z\times_{f}G$ and
\begin{equation*}
 f'(g,h):= \sigma(g)\sigma(h)\sigma(gh)^{-1}
\end{equation*}
defines a $Z$-valued cocycle. Since $Z$ acts freely on $Z\times _{f}G$, we have
$(0,g)=\sigma(g)\cdot b(g)$ for $b\from G\to Z$, smooth on a identity neighbourhood and satisfying \eqref{eqn:cocycleCoboundaryaCondition}.
The ``only if'' part is clear from the construction of $\tau(f)$.
\end{proof}

We thus obtain a sequence
\begin{equation*}
 0\to H^{2}_{\op{glob}}(G,Z)\to H^{2}_{\op{loc}}(G,Z)\xrightarrow{\tau} \check{H}^{1}(G,Z)
\end{equation*}
which is obviously exact. It would be interesting to determine for which groups and which
coefficients the map $\tau$ is \emph{not} surjective. Note that the case of $Z$ being
connected is the interesting one, since for a discrete group $A$, each principal
$A$-bundle over $G$ is a covering and thus admits a compatible Lie group structure.

\section{The universal central extension of Loop groups}

The results described in the preceding section applies to loop groups $\Omega K$
in the following way. If $\sprod\from \fk\times \fk\to \R$ denotes the Killing
form (which is non-degenerate and negative definite in our case), then
\begin{equation*}
 \omega\from \Omega \fk\times \Omega\fk \to \R,\quad (f,g)\mapsto 
 \int_{S^{1}} \langle f(t),g'(t)\rangle dt
\end{equation*}
defines a continuous Lie algebra cocycle. If we normalise
$\sprod$ in such way that the left-invariant extension $\omega^{l}$ of $\omega$
satisfies $\int_{\sigma}\omega^{l}=1$ for $\sigma$ a generator\footnote{The ambiguity in the sign that one still has for the normalisation of $\sprod$ will play no role in the sequel.} of $\pi_{2}(\Omega K)\cong \pi_{3}(K)\cong \Z$,
then the calculations in
\cite{MaierNeeb03Central-extensions-of-current-groups} show that for $Z=\R$ we have
$\per_{\omega}(\pi_{2}(\Omega K))=\Z$. Thus the cocycle $k\cdot \omega$ integrates to a
locally smooth group cocycle $f_{k}\from \Omega K\times \Omega K\to U(1)$, defining
a central extension
\begin{equation*}
 U(1)\to \wh{\Omega K}_{k}\to \Omega K
\end{equation*}
if and only if $k\in \Z$. Moreover, the central extension for $k=\pm 1$ is
universal, as it is shown in
\cite{MaierNeeb03Central-extensions-of-current-groups}. This means that for each other central extension $Z\to \wt{\Omega K}\to \Omega K$ there exist unique morphisms
$U(1)\to Z$ and $\wh{\Omega K}_{\pm 1}\to \wt{\Omega K}$ making the diagram
\begin{equation*}
\begin{CD}
 U(1) @>>> \wh{\Omega K}_{\pm 1}@>>> \Omega K \\
 @VVV @VVV @VVV\\
 Z @>>> \wh{\Omega K}@>>> \Omega K
\end{CD}
\end{equation*}
commute.

There also exist more ad-hoc constructions of $\wh{\Omega K}_{k}$, cf.\ \cite{PressleySegal86Loop-groups}, \cite{Mickelsson87Kac-Moody-groups-topology-of-the-Dirac-determinant-bundle-and-fermionization},
\cite{Murray88Another-construction-of-the-central-extension-of-the-loop-group}, \cite{MurrayStevenson01Yet-another-construction-of-the-central-extension-of-the-loop-group} or \cite{MurrayStevenson03Higgs-fields-bundle-gerbes-and-string-structures} which, more or less, all construct $\wh{\Omega K}_{k}$
first constructing a central extension $\wh{P_{e}\Omega K}\to P_{e}\Omega K$, corresponding to the pull-back of $\wh{\Omega K}_{k}$
along $\ev\from P_{e}\Omega K\to \Omega K$, and then considering an appropriate quotient of $\wh{P_{e}\Omega K}$. Since
$P_{e}\Omega K$ is contractible, the results of the preceding sections imply that the pull-back of the central extension of $L(\Omega K)\cong \Omega \fk$ along the evaluation homomorphism
$L(\ev)\from L(P_{e}\Omega K)\to L\fk$ to $L(P_{e}\Omega K)\cong P_{e} \Omega\fg$ integrates to a central extension
\begin{equation*}
 U(1)\to \wh{P_{e}\Omega K}\to P_{e}\Omega K,
\end{equation*}
given by a globally smooth cocycle $\kappa\from PG\times PG\to U(1)$.
The constructions of $\wh{\Omega K}$ cited above all deal with an explicit description of a normal subgroup $N_{k} \unlhd \wh{P_{e}\Omega K}$ in order to obtain an induced central extension
\begin{equation*}
  U(1)\to \underbrace{\wh{P_{e}\Omega K} / N_{k}}_{\cong\wh{\Omega K}_{k}} \to \underbrace{P_{e}\Omega K/\Omega (\Omega K)}_{\cong \Omega K}
\end{equation*}
(cf.\ also \cite[Section III]{GlocknerNeeb03Banach-Lie-quotients-enlargibility-and-universal-complexifications}). 

\section{Central extensions of loop groups from Lie groupoids}

We shall put more structure on the ad-hoc construction of $\wh{\Omega K}$ from the
previous section. In particular, we show that $\wh{P_{e}\Omega K} / N_{k}$ may be
obtained as the quotient of an action Lie groupoid, which also exists for
non-integral values of $k$.

\begin{remark}
 We briefly recall the essential notions for Lie groupoids. A Lie
 groupoid is a category object in the category of locally convex
 manifolds, such that source and target maps admit local
 inverses\footnote{We shall use concepts from the usual theory of Lie
 groupoids by replacing  the term ``surjective submersion'' at each
 occurrence by the term ``admits local inverses''. This is equivalent in
 the finite-dimensional case but may not be in the infinite-dimensional
 one.}. More precisely, it consists of two locally convex manifolds
 $M_{1}$ and $M_{0}$, together with smooth maps
 $\id\from M_{0}\to M_{1}$ and $s,t\from M_{1}\to M_{0}$, admitting
 local inverses, and a smooth map\footnote{Note that the existence of
 local inverses for $s$ and $t$ ensures the existence of a manifold
 structure on $M_{1\,\,s\!\!}\times_{t} M_{1}$.}
 $\circ \from M_{1\,\,s\!\!}\times_{t} M_{1}\to M_{0}$ satisfying the usual
 relations of identity, source, target and composition map of a small
 category. Moreover, we require that each morphism of this category is
 invertible an that the map $M_{1}\to M_{1}$, assigning to each morphism
 its inverse, is smooth. The quotient of such a Lie groupoid is defined
 to be the set of equivalence classes of isomorphic objects. The smooth
 structure on $M_{0}$ may or may not induce a smooth structure on the
 quotient, depending on how badly the quotient actually is behaved.

 A typical example of a Lie groupoid, called \emph{action groupoid}, is
 obtained from a smooth right action of a Lie group $G$ on a manifold
 $M$. With this data given, we set $M_{0}:=M$, $M_{1}:=M\times G$,
 $\id(m):=(m,e)$, $s(m,g):=m$, $t(m,g):=m.g$ and
 $(m.g,h)\circ(m,g):=(m,g\cdot h)$. Clearly, the inverse of $(m,g)$ is
 $(m.g,g^{-1})$. The quotient of the Lie groupoid clearly is given by
 $M/G$. If it admits a smooth structure such that the quotient map
 $M\to M/G$ is smooth an admits local inverses, then the action groupoid
 is Morita equivalent\footnote{Morita equivalent Lie groupoids are the
 correct replacement for the concept of equivalent categories. In fact,
 Morita equivalent Lie groupoids are equivalent as categories and
 possess the same amount of ``differential'' information, cf.\
 \cite{MoerdijkMrcun03Introduction-to-foliations-and-Lie-groupoids}.} to
 the Lie groupoid with $M_{0}=M_{1}=M/G$ and all structure maps the
 identity. If $M/G$ does not carry a smooth structure, then the action
 groupoid is an appropriate replacement for $M/G$.
\end{remark}

In order to motivate our procedure
we recall that $N_{k}$ is defined to be the subset
\begin{equation*}
\{(z,\gamma)\in U(1)\times P_{e}\Omega K:\gamma\in C_{*}^{\infty}(S^{1},\Omega K), z=\exp(-k\cdot \int_{D_{\gamma}}\omega^{l})\},
\end{equation*}
where $D_{\gamma}\from B^{2}\to \Omega K$ is a smooth map with
$\left.D_{\gamma}\right|_{\partial B^{2}}=\gamma$ and $\omega^{l}$ is the left-invariant
2-from on $\Omega K$ with $\omega^{l}(e)=\omega$. Since $\omega^{l}$ is an
integral 2-from on $\Omega K$, the value of $\exp(k\cdot\int_{D_{\gamma}}\omega^{l})$ does not
depend on the choice of $D_{\gamma}$ if $k\in \Z$.
The groupoid that we will construct carries some more
information, namely not only the boundary value of $D_{\gamma}$ but also the homotopy type of
it relative to $\partial B^{2}$. This information is contained in the group
\begin{equation*}
 C^{\infty}_{*}(B^{2},\Omega K)/C_{*}^{\infty}(S^{2},\Omega K)_{e},
\end{equation*}
(where we identify $C_{*}^{\infty}(S^{2},\Omega K)$ with the normal subgroup in
$C^{\infty}_{*}(B^{2},\Omega K)$ of functions that vanish on $\partial B^{2}$) which we shall
now endow with a Lie group structure. The following proof shall make use of the fact that
smooth and continuous homotopies of functions with values in locally convex manifold agree,
we refer to \cite{Wockel06A-Generalisation-of-Steenrods-Approximation-Theorem} for details on
this.

\begin{lemma}
If $G$ is a connected locally convex Lie group with Lie algebra $\fg$, then the quotient group
\begin{equation*}
 C_{*}^{\infty}(B^{2},G)/ C_{*}^{\infty}(S^{2},G)_{e}
\end{equation*}
carries a Lie group structure, modelled on $C^{\infty}_{*}(S^{1},\fg)$.
\end{lemma}

\begin{proof}
 We shall make use of the Lie group structure on $C^{\infty}_{*}(M,G)$ (with respect to
 point-wise group operations), which exists for each compact manifold $M$, possibly with
 corners \cite{Wockel06Smooth-extensions-and-spaces-of-smooth-and-holomorphic-mappings}.
 If $U$ is open in $G$, then
 \begin{equation*}
  C^{\infty}(M,U):=\{f\in C^{\infty}_{*}(M,G):{f}(M)\se U\}
 \end{equation*}
 is open
 in $C^{\infty}(M,G)$ and, likewise, if $U'$ is open in $\fg$, then $C^{\infty}(M,U')$
 is open in $C^{\infty}(M,\fg)$.
 If $U\se G$ is an open identity neighbourhood and $\varphi\from U\to \varphi(U)\se \fg$
 is a chart with $\varphi(U)$ open and convex and satisfying $\varphi(e)=0$, then
 a chart for the manifold structure, underlying $C^{\infty}_{*}(M,G)$, is given by
 \begin{equation*}
 C^{\infty}(M,U) \ni {f}\mapsto
 \varphi \circ{f}\in C^{\infty}(M,\varphi(U)).
 \end{equation*}
 Clearly, this induces a map
 \begin{equation*}
  \wt{\varphi}\from q(C_{*}^{\infty}(B^{2},U)) \to C_{*}^{\infty}(S^{1},\varphi(U)),
 \quad [{f}]\mapsto \varphi \circ\left(\left.{f}\right|_{\partial B^{2}}\right),
 \end{equation*}
 where $q\from C_{*}^{\infty}(B^{2},G)\to C_{*}^{\infty}(B^{2},G)/C^{\infty}_{*}(S^{2},G)_{e}$
 denotes the canonical quotient map. This map is bijective since each map $f\in C^{\infty}_{*}(S^{1},\varphi(U))$ has a 
homotopy to the map which is constantly $0$, defining an extension of $f$
to a map $F\from B^{2}\to \varphi(U)$ with $\left.F\right|_{\partial B^{2}}=f$ and
$[\varphi^{-1}\circ F]$ is mapped to $f$ under $\wt{\varphi}$. Similarly, we deduce that
$\wt{\varphi}$ is injective, since each two maps in $C_{*}^{\infty}(B^{2},U)$, which restrict to the same value on $\partial B^{2}$, are homotopic.

We are now ready to verify that the conditions of Theorem \ref{thm:globalisation-of-smooth-structures-on-groups} are satisfied, which we want to apply to the subset $W:=q(C_{*}^{\infty}(B^{2},U))$. On this we have a smooth structure, induced by the bijection $\wt{\varphi}$. Moreover, if $V\se U$ is an open
identity neighbourhood of $G$ with $V^{2}\se U$ and $V^{-1}=V$, then $q(C_{*}^{\infty}(B^{2},V))$ is open in $q(C_{*}^{\infty}(B^{2},U))$. The structure maps on the mapping group under consideration are all given
by the point-wise group structure in $G$, and so it follows that the coordinate representation
of the structure maps on $q(C_{*}^{\infty}(B^{2},V))$
coincides with the coordinate representation of the structure maps of
$C^{\infty}_{*}(S^{1},G)$. Since the latter are smooth it follows that the structure maps
on $q(C_{*}^{\infty}(B^{2},V))$ are smooth. Finally, 
$q(C_{*}^{\infty}(B^{2},V))$ generates $C_{*}^{\infty}(B^{2},G)/C^{\infty}_{*}(S^{2},G)_{e}$,
 because $G$ is connected.
\end{proof}

Note that there is a natural homomorphism
\begin{equation*}
\cK:= C_{*}^{\infty}(B^{2},\Omega K)/ C_{*}^{\infty}(S^{2},\Omega K)_{e}\to 
 C^{\infty}_{*}(S^{1},\Omega K),\quad [f]\mapsto \left.f\right|_{\partial B^{2}},
\end{equation*}
which obviously is smooth and surjective, because $\pi_{1}(\Omega K)$ vanishes. The kernel
of this map is $C_{*}^{\infty}(S^{2},\Omega K)/ C_{*}^{\infty}(S^{2},\Omega K)_{e}\cong \pi_{2}(\Omega K)$ and we thus obtain a central extension
\begin{equation*}
 \pi_{2}(\Omega K)\to \cK\to C_{*}^{\infty}(S^{1},\Omega K).
\end{equation*}
That $\pi_{2}(\Omega K)$ is in fact central follows from the fact that it is a discrete
normal subgroup of the connected group $\cK$.
For general, not necessarily simply connected $G$, we only obtain a crossed module
\begin{equation*}
 C_{*}^{\infty}(B^{2},G)/ C_{*}^{\infty}(S^{2},G)_{e}\to C_{*}^{\infty}(S^{1},G).
\end{equation*}
Since the image of this morphism is precisely $C_{*}^{\infty}(S^{1},G)_{e}$, this in turn gives
rise to the four term exact sequence
\begin{equation*}
 \pi_{2}(G)\to C_{*}^{\infty}(B^{2},G)/ C_{*}^{\infty}(S^{2},G)_{e}\to C_{*}^{\infty}(S^{1},G)\to \pi_{1}(G).
\end{equation*}
This sequence has a characteristic class in in $H^{3}(\pi_{1}(G),\pi_{2}(G))$, which has first been constructed in
\cite{EilenbergMacLane46Determination-of-the-second-homology-and-cohomology-groups-of-a-space-by-means-of-homotopy-invariants}.

The second smooth map, naturally associated to $\cK$ is given by
$\cK\to \R$, $ [f]\mapsto \int _{f}\omega^{l}$,
where the integral only depends on the homotopy class of $f$ because $\omega^{l}$ is
closed.

For the following lemma we define a generalisation of the
the cocycle $\kappa$ by
\begin{multline*}
\kappa_{k}\from (P_{e}\Omega K)\times (P_{e}\Omega K)\to U(1),\quad\\
(\gamma,\eta)\mapsto
\exp\left(k\cdot \int _{0}^{1}\int_{0}^{1}\langle \gamma(s)^{-1}\gamma'(s),\eta'(t)\eta(t)^{-1}\rangle\, ds\, dt\right),
\end{multline*}
which is for $k=1$ the cocycle $\kappa$ from \cite{Murray88Another-construction-of-the-central-extension-of-the-loop-group} (cf.\ also \cite{BaezCransStevensonSchreiber07From-loop-groups-to-2-groups}).

\begin{proposition}
 For each $k\in \R$, the Lie group $\cK$ acts smoothly from the right
 on $U(1)\times P_{e}\Omega K$ by
 \begin{equation}\label{eqn:actionForActionGroupoid}
  (z,\gamma).[f]:=( z\cdot\exp(-k\cdot\int_{f}\omega^{l})\cdot \kappa_{k}(\gamma,\left.f\right|_{\partial B^{2}}),\gamma\cdot \left.f\right|_{\partial B^{2}}).
 \end{equation}
\end{proposition}

\begin{proof}
 It is clear that the action map is smooth on the product
 \begin{equation*}
  (U(1)\times P_{e}\Omega K)\times\cK,
 \end{equation*}
 because the restriction map $\cK\to C^{\infty }_{*}(S^{1},\Omega K)$ and
 the integration map $\cK\to U(1)$ are smooth.

 In order to show that \eqref{eqn:actionForActionGroupoid} actually
 defines a group action we have to verify that
 $(z,\gamma).[f\cdot g]=((z,\gamma).[f]).[g	]$, which is equivalent
 to
 \begin{multline*}
  \exp(-k\cdot\int_{f\cdot g}\omega^{l})\cdot \kappa_{k}(\gamma,\left.(f\cdot g)\right|_{\partial B^{2}})=\\
  \exp(-k\cdot\int_{f}\omega^{l})\cdot \kappa_{k}(\gamma,\left.f\right|_{\partial B^{2}})\cdot
  \exp(-k\cdot\int _{g}\omega^{l})\cdot \kappa_{k}(\gamma\cdot\left.f\right|_{\partial B^{2}},\left.g\right|_{\partial B^{2}}).
 \end{multline*}
 But this in turn follows immediately from the cocycle condition for
 $\kappa_{1}$, because
 \begin{equation*}
  \kappa_{1}(\left.f\right|_{\partial B^{2}},\left.g\right|_{\partial B^{2}})=\exp(\int _{f}\omega^{l})\cdot\exp(\int _{g}\omega^{l})\cdot\exp(-\int _{f\cdot g}\omega^{l})
 \end{equation*}
 follows from $\left.f\right|_{\partial B^{2}},\left.g\right|_{\partial B^{2}}\in C_{*}^{\infty}(S^{1},\Omega K)$ (cf.\ \cite[Sect.\
 6]{Murray88Another-construction-of-the-central-extension-of-the-loop-group}).
\end{proof}

For each $k\in \R$, the action \eqref{eqn:actionForActionGroupoid} now defines an action
Lie groupoid
\begin{equation}\label{eqn:actionGroupoid}
 \big(U(1)\times P_{e}\Omega K\times \cK  \rrarrow[k] U(1)\times P_{e}\Omega K\big),
\end{equation}
i.e., $s(z,\gamma,[f])=(z,\gamma)$, $t(z,\gamma,[f])=(z,\gamma).[f]$ and
\begin{equation*}
((z,\gamma).[f],[f']) \circ ((z,\gamma),[f])=(z,\gamma,[f\cdot f']).
\end{equation*}
From formula \eqref{eqn:actionForActionGroupoid} we see in particular, that the action of
$\cK$, and thus the Lie groupoid \eqref{eqn:actionGroupoid} is not proper for arbitrary
$k$.
In fact, the subgroup $\Z\cong\pi_{2}(\Omega K)\se \cK$
acts on $U(1)\times P_{e}\Omega K$ by
\begin{equation*}
 a.(z,\gamma)=(z\cdot \exp(-k\cdot a),\gamma),
\end{equation*}
since we assumed that $\omega$ is normalised so that $\int_{\sigma}\omega^{l}=1$ for a generator $[\sigma]$ of $\pi_{2}(\Omega K)\cong \pi_{3}(K)$. Of course, the interesting range for $k$ in the previous proposition is $k\in [0,1]$, for then the quotient of the action ``interpolates'' between the trivial and the universal extension:
\begin{align*}
 \tx{for $k=0$}&:\,\, U(1)\times P_{e}\Omega K/\cK \cong U(1)\times \Omega K\\
 \tx{for $k=1$}&:\,\, U(1)\times P_{e}\Omega K/\cK \cong \wh{\Omega K}_{1}\\
\end{align*}
Moreover, we see that for each $k\in \Q$ the quotient $U(1)\times P_{e}\Omega K/\cK$
exists as a manifold, we shall give a precise argument for this at the end of the next
section. However, the group structure on $U(1)\times P_{e}\Omega K$, given by the
cocycle $\kappa$, only induces a group structure on the quotient in the case $k\in \Z$.

\section{Lie groupoids as principal 2-bundles}

In this section we show that the Lie groupoids, derived in the previous section, possess the
structure of a principal 2-bundle. For this we give at first a very short and condensed 
introduction to principal 2-bundles. The details can be found in \cite{Wockel09Principal-2-bundles-and-their-gauge-2-groups}.

A strict Lie 2-group is a category object in the category of locally convex Lie groups i.e., it consists of two locally convex Lie groups $G_{0}$ and $G_{1}$, together with morphisms $s,t\from G_{1}\rrarrow G_{0}$, a
morphism $i\from G_{0}\to G_{1}$ and a morphism $c\from {G_{1}}_{\,s\!\!}\times_{t}G_{1}\to G_{1}$
(assuming that the pull-back ${G_{1}}_{\,s\!\!}\times_{t}G_{1}$ exists), such that
$(G_{0},G_{1},s,t,i,c)$ constitutes a small category. In short, we write $(G_{1}\rrarrow G_{0})$ for this (cf.\ \cite{BaezLauda04Higher-dimensional-algebra.-V.-2-groups} and \cite{Porst08Strict-2-Groups-are-Crossed-Modules}).
A smooth 2-space  is simply a Lie
groupoid and similar to the case of Lie groups and manifolds, one defines a (right) 
$(G_{1}\rrarrow G_{0})$-2-space to be a 2-space $(M_{1}\rrarrow M_{0})$, together with a smooth functor
\begin{equation*}
(\rho_{1}, \rho_{0})\from (M_{1}\rrarrow M_{0})\times (G_{1}\rrarrow G_{0})\to (M_{1}\rrarrow M_{0}),
\end{equation*}
such that $\rho_{1}$ defines a (right) $G_{1}$-action on $M_{1}$ and $\rho_{0}$ defines a (right) $G_{0}$-action on on $M_{0}$. Similarly, one defines a morphism of $(G_{1}\rrarrow G_{0})$-2-spaces $(M_{1}\rrarrow M_{0})$ and $(N_{1}\rrarrow N_{0})$ to be a smooth functor
$(\varphi_{1}\times \varphi_{0})\from (M_{1}\times M_{0})\to (N_{1}\times N_{0})$ such
that $\varphi_{1}$ (respectively $\varphi_{0}$) defines a morphism of $G_{1}$ (respectively $G_{0}$)-spaces. A 2-morphism $\alpha \from\varphi\Rightarrow \psi$ between two morphisms $\varphi,\psi\from (M_{1}\times M_{0})\to (N_{1}\times N_{0})$ of $(G_{1}\rrarrow G_{0})$-spaces consists of a smooth map $\alpha\from M_{0}\to N_{1}$ such that
$\alpha$ defines a natural transformation between the functors $\varphi$ and $\psi$ and, moreover, satisfies $\alpha(m.g)=\alpha(m).\id_{g}$ for each $m\in M_{0}$ and $g\in G_{0}$.

With this said one defines a principal $(G_{1}\rrarrow G_{0})$-2-bundle over the smooth
manifold $M$ (viewed as a smooth 2-space with only identity morphisms, we write $\ul{M}$ for
this 2-space) as follows. It is a smooth $(G_{1}\rrarrow G_{0})$-2-space $(P_{1}\rrarrow P_{0})$, together
with a smooth functor $\pi\from (P_{1}\rrarrow P_{0})\to (M\rrarrow M)$, commuting with the
action functor $\rho$, such that there exist
\begin{itemize}
	\item an open cover $({U_{i}})_{i\in I}$ of $M$
	\item morphisms
	\begin{align*}
	\Phi_{i}&\from \pi^{-1}(\ul{U_{i}})\to \ul{U_{i}}\times (G_{1}\rrarrow G_{0})\quad\tx{and}\\
	\ol{\Phi}_{i}&\from  \ul{U_{i}}\times (G_{1}\rrarrow G_{0})\to \pi^{-1}(\ul{U_{i}})
	\end{align*}	
	of $(G_{1}\rrarrow G_{0})$-2-spaces
	\item 2-morphisms
	\begin{align*}
	\tau_{i}&\from \Phi_{i}\circ \ol{\Phi}_{i}\Rightarrow 
	   \id_{\ul{U_{i}}\times (G_{1}\rrarrow G_{0})}\\
	\ol{\tau}_{i}&\from \ol{\Phi}_{i}\circ \Phi_{i}\Rightarrow \id_{\pi^{-1}(\ul{U_{i}})}
	\end{align*}
	between morphisms 	of $(G_{1}\rrarrow G_{0})$-2-spaces,
\end{itemize}
such that $\pi$, $\Phi_{i}$ and $\ol{\Phi}_{i}$ commute in the usual way with the projection
functor $\pr\from \ul{U}_{i}\times (G_{1}\rrarrow G_{0})\to \ul{U_{i}}$.

We are now aiming at showing that the action Lie groupoid
\begin{equation*}
 (U(1)\times P_{e}\Omega K\times \cK  \rrarrow[k] U(1)\times P_{e}\Omega K)
\end{equation*}
possesses the structure of a principal 2-bundle (we used the subscript $_{k}$ to denote the
value of $k$ in the action map \eqref{eqn:actionForActionGroupoid}). The structure 2-group of
this bundle shall be given by $(U(1)\times \pi_{2}(\Omega K)\rrarrow[k] U(1))$ with
$s(z,[\sigma])=z$, $t(z,[\sigma])=z\cdot \exp(-k\cdot\int_{\sigma} \omega ^{l})$ and $(z\cdot
\exp(-k\cdot\int_{\sigma}\omega^{l}),[\sigma'])\circ(z,[\sigma])=(z,[\sigma'\cdot \sigma])$.\\

Before showing the claim of this section, we have to pass from the action Lie groupoid
\eqref{eqn:actionGroupoid} to a Morita equivalent one, which we will denote by $(P_{1}\rrarrow[k] P_{0})$. For this we choose a system
$(\sigma_{i}\from U_{i}\to P_{e}\Omega K)_{i\in I}$ of smooth local sections of the principal bundle $\ev\from P_{e}\Omega K\to \Omega K$. For technical reasons, that will become apparent later, we choose this system so that there exists smooth maps $\sigma_{ij}\from U_{i}\cap U_{j}\to \cK$ such that $\sigma_{i}(x)=\sigma_{j}(x)\cdot \left.\sigma_{ij}(x)\right|_{\partial B^{2}}$. Then we set
\begin{equation*}
 P_{0}:=\coprod_{i\in I} \left(U(1)\times\{\sigma_{i}(x):x\in U_{i}\}\right),
\end{equation*}
which we endow with the smooth structure induced from $ U(1)\times P_{e} \Omega K$. The set of morphisms
we set to be
\begin{equation*}
P_{1}:=\{(z,\gamma,\eta,[f])\in U(1)\times P_{0}\times P_{0}\times \cK: \ev(\gamma)=\ev(\eta), \gamma=\eta\cdot \left.f\right|_{\partial B^{2}}\}.
\end{equation*}
For a fixed choice of $\gamma$ and $\eta$, the possible different choices of $[f]$ are parametrised  by $\pi_{2}(\Omega K)$, and so $P_{1}$ has a natural manifold structure,
modelled on $C_{*}^{\infty}(S^{1},\Omega K)$. Source and target maps are induced by the two
projections from $P_{1}$ to $P_{0}$ and composition is induced by multiplication in $\cK$.

We define a smooth functor from 
$(P_{1}\rrarrow[k] P_{0})$ to $(U(1)\times P_{e}\Omega K\times \cK  \rrarrow[k] U(1)\times P_{e}\Omega K)$ by inclusion on objects and on morphisms by $(z,\gamma,\eta,[f])\mapsto (z,\gamma,[f])$. One easily checks that this functor actually defines a Morita equivalence.
\\

There exists an obvious $(U(1)\times \pi_{2}(\Omega K)\rrarrow[k] U(1))$-2-space structure on $(P_{1}\rrarrow[k] P_{0})$, given by
\begin{alignat}{3}
(z,\gamma).w&=(z\cdot w,\gamma)  &&\quad\tx{ on objects and by} \label{eqn:action1}\\
(z,\gamma,\eta,[f]).(w,[\sigma])&=(z\cdot w\cdot \exp(-k\cdot\int_{\sigma}),\gamma,\eta,[f\cdot \sigma])&&\quad\tx{ on morphisms.}\label{eqn:action2}
\end{alignat}
Moreover, there exists a natural smooth functor $\pi\from (P_{1}\rrarrow[k] P_{0})\to \ul{\Omega K}$,
given on objects by $(z,\gamma)\mapsto \ev(\gamma)$ and on morphisms by $(z,\gamma,\eta,[f])\mapsto \ev({\gamma})$.
We are now ready to prove the main result on this section.

\begin{proposition}
 The $(U(1)\times \pi_{2}(\Omega K)\rrarrow[k] U(1))$-2-space structure
 on $(P_{1}\rrarrow [k] P_{0})$, given by \eqref{eqn:action1} and
 \eqref{eqn:action2}, along with the smooth functor $\pi$, defines a principal 2-bundle.
\end{proposition}

\begin{proof}
 We observe that $(z,\gamma)$ is an
 object of $\pi^{-1}(\ul{U_{i}})$ if and only if $\ev(\gamma)\in U_{i}$
 and $\gamma=\sigma_{i}(\ev(\gamma))$. From
 $\ev(\gamma)=\ev(\gamma\cdot\left.f\right|_{\partial B^{2}})$ for each
 $[f]\in \cK$ it follows that a morphisms has source in $\pi^{-1}(\ul{U_{i}})$ if
 and only if it has target in $\pi^{-1}(\ul{U_{i}})$, so that
 $\pi^{-1}(\ul{U_{i}})$ is in fact a full subcategory.

 We now define local trivialisations $\Phi_{i}$ by
 \begin{alignat*}{3}
  (z,\gamma)&\mapsto (\ev(\gamma),z) &&\quad\tx{on objects and by}\\
  (z,\gamma,\eta,[f])&\mapsto(z,\sigma_{ij}(\gamma,\eta)\cdot[f]^{-1})&&\quad\tx{on morphisms.}
 \end{alignat*}
 This is smooth due to the requirements that we put on the choice of
 $(\sigma_{i}\from U_{i}\to P_{e}\Omega K)_{i\in I}$ and that it
 actually defines a functor follows from the fact that
 $\pi_{2}(\Omega K)$ is central in $\cK$. The ``inverse''
 trivialisations $\ol{\Phi}_{i}$ we define by
 \begin{alignat*}{3}
  (l,z)&\mapsto (z,\sigma_{i}(l))&&\quad\tx{on objects and by}\\
  (l,(z,[\sigma]))&\mapsto (z,\sigma_{i}(l),\sigma_{i}(l),e)&&\quad\tx{on morphisms.}
 \end{alignat*}
 These obviously define smooth functors commuting with the
 $(U(1)\times \pi_{2}(\Omega K)\rrarrow[k] U(1))$-action, and we
 have $\Phi_{i}\circ \ol{\Phi}_{i}=\id$. We then define
 $\ol{\tau}_{i}\from \ol{\Phi}_{i}\circ \Phi_{i}\Rightarrow \id$ by
 \begin{equation*}
  (z,\gamma)\mapsto (z,\sigma_{i}(x),\sigma_{j}(x),\sigma_{ij}(x)) \tx{ if }\gamma=\sigma_{j}(x)\tx{ for }x\in U_{ij}.
 \end{equation*}
 It is easily checked that $\ol{\tau}_{i}$ actually defines a natural
 transformation and satisfies $\ol{\tau}_{i}((z,\gamma).z')=\ol{\tau}_{i}((z,\gamma)).(z',e)$.
\end{proof}

Note that the functors $\Phi_{i}$ and the natural transformations $\ol{\tau}_{i}$ in the
previous proof were smooth for they only need to be defined if $(\gamma,\eta)$ can be written
as $(\sigma_{i}(x),\sigma_{j}(x))$ for $x=\ev(\gamma)=\ev(\eta)$. If one tried to define a
2-bundle structure on the whole action groupoid \eqref{eqn:actionGroupoid} in a similar way,
then one would need a smooth global section of $\cK\to C_{*}^{\infty}(S^{1},\Omega K)$, which
does not exist. Thus the passage to the Morita equivalent groupoid $(P_{1}\rrarrow[k] P_{0})$
was necessary to ensure the smoothness properties of the local trivialisations.

\begin{corollary}
 If $k\in \Q$, then the quotient $P_{k}$ of the groupoid
 $(P_{1}\rrarrow[k] P_{0})$ can be endowed with the structure of a
 smooth manifold. Moreover, the action \eqref{eqn:action1} induces on
 $P_{k}$ the structure of a smooth $(U(1)/k)$-principal bundle.
\end{corollary}

\begin{proof}
 This is exactly the construction of the band of a principal 2-bundle from \cite{Wockel09Principal-2-bundles-and-their-gauge-2-groups}.
\end{proof}

The previous result can also be obtained as in Section
\ref{sect:theTopologicalTypeOfCentralExtensions} by considering the Lie group
$U(1)/k=\R/(\Z+k\Z)$. This shows actually that $P_{k}$ can also be endowed with a Lie group
structure, turning \begin{equation*} (U(1)/k)\to P_{k}\to \Omega K \end{equation*} into a
central extension of Lie groups. However, the group structure on $P_{k}$ is not induced by the
one on $U(1)\times_{\kappa_{1}}P_{e}\Omega K$ any more.

\bibliographystyle{amsalpha}
\def\polhk#1{\setbox0=\hbox{#1}{\ooalign{\hidewidth
  \lower1.5ex\hbox{`}\hidewidth\crcr\unhbox0}}} \def\cprime{$'$}
\providecommand{\bysame}{\leavevmode\hbox to3em{\hrulefill}\thinspace}
\providecommand{\MR}{\relax\ifhmode\unskip\space\fi MR }
\providecommand{\MRhref}[2]{%
  \href{http://www.ams.org/mathscinet-getitem?mr=#1}{#2}
}
\providecommand{\href}[2]{#2}

\end{document}